\documentclass[a4paper,10pt]{article}
\usepackage[margin = 0.9 in]{geometry}
\usepackage{amsfonts,amsmath,amsthm,amssymb,indentfirst,placeins,float,graphicx,booktabs, verbatim,bbm,microtype,authblk,multicol}

\usepackage[hyperindex,breaklinks]{hyperref}
\hypersetup{linkcolor=blue, citecolor=red, colorlinks=true}
\usepackage[T1]{fontenc}
\usepackage{lmodern}
\pdfoutput=1

\theoremstyle{plain}
\newtheorem{myth}{Theorem}[section]
\newtheorem{mylm}[myth]{Lemma}
\newtheorem{mypr}[myth]{Proposition}
\newtheorem{mycor}[myth]{Corollary}

\theoremstyle{remark}
\newtheorem{mydef}[myth]{Definition}
\newtheorem{myass}[myth]{Assumption}

\newcommand{\ds}{\displaystyle}
\newcommand{\mc}{\mathcal}
\newcommand{\ul}{\underline}
\newcommand{\ol}{\overline}
\newcommand{\expoc}{E_{t,z,x}}

\newcommand{\x}{X^{u,t,x}}
\newcommand{\z}{Z^{u,t,z,x}}
\numberwithin{equation}{section}

\DeclareMathOperator*{\esssup}{ess\,sup}
\DeclareMathOperator*{\essinf}{ess\,inf}

\begin{document}
\title{\textbf{A Stochastic Control Approach to Bid-Ask Price Modelling}}
\author[1]{Engel John C. Dela Vega}
\author[1,2]{Robert J. Elliott\footnote{Corresponding Author; Email: relliott@ucalgary.ca}}
\affil[1]{UniSA Business, University of South Australia, Adelaide, SA 5000, Australia}
\affil[2]{Haskayne School of Business, University of Calgary, Calgary, Alberta, Canada T2N 1N4}
\date{}
	\maketitle
	\vspace{-1cm}
	\begin{abstract}
		This paper develops a model for the bid and ask prices of a European type asset by formulating a stochastic control problem. The state process is governed by a modified geometric Brownian motion whose drift and diffusion coefficients depend on a Markov chain. A Girsanov theorem for Markov chains is implemented for the change of coefficients, including the diffusion coefficient which cannot be changed by the usual Girsanov theorem for Brownian motion. The price of a European type asset is then determined using an Esscher transform and a system of partial differential equations. A dynamic programming principle and a maximum/minimum principle associated with the stochastic control problem are then derived to model bid and ask prices. These prices are not quotes of traders or market makers but represent estimates in our model on which reasonable quantities could be traded.
		
		\vspace{12pt}
		
	\end{abstract}

	\section{Introduction}
	\vspace{-0.5cm}
	Stochastic optimal control problems involving Markov processes have been extensively studied \cite{bertsekas:measurableselection, elliott:hmm, soner:controlledprocess, gihman:controlledprocess, zhenting:controlledprocess}. The most common methods in solving these control problems are the dynamic programming principle and the stochastic maximum principle. Stochastic modelling under a Markovian regime-switching environment incorporates jumps in prices due to changes in the state of the economy, which is modelled by a Markov chain. In the context of stochastic control theory, regime-switching state equations have also been investigated. Maximum principles of regime-switching state equations are developed in several papers \cite{ref1:opt,ref2:opt,elliott:stochmaxprin}. A dynamic programming principle under regime-switching is discussed in \cite{azevedo:control}.
	
	This paper aims to discuss bid and ask prices by formulating a control problem. These two prices are modelled as different interpretations of dynamics for the asset by supposing the coefficients in the dynamics are functions of a continuous-time finite state Markov chain whose evolution depends on parameters which can be thought of as control variables. In turn, these parameters give rise to a family of probability measures which can be considered as possible future scenarios. 
	
	
	The paper is organized as follows: the price process dynamics are introduced in Section \ref{ustart} as a modified geometric Brownian motion where the drift and diffusion coefficients depend on the Markov chain, together with a jump term. Section \ref{controlsection} introduces controlled Markov chains. The chain is then controlled by the first change of measure, discussed in Section \ref{change of measure for markov chains}, which allows changes in the diffusion coefficient by changing the rate matrix of the chain. 
	The second measure change in Section \ref{esscher} involves the use of the Esscher transform to change the drift coefficient and to find risk-neutral dynamics for the underlying process. The price of a European type asset is then determined using a system of partial differential equations and the homotopy analysis method in Section \ref{uend}. Finally, the bid and ask prices are characterized as solutions of optimal control problems in Section \ref{bid-ask price models}.

	\section{Regime-Switching Process}\label{ustart}
	\vspace{-0.2cm}
	Consider a complete probability space $(\Omega,\mc{F},P)$, where $P$ is a real-world probability measure. Write $\mc{T}:=[0,T]$. Let $\mathbf{W}:=\{W_t\}_{t\in \mc{T}}$ be a standard Brownian motion on $(\Omega,\mc{F},P)$. Assume that states in an economy are modelled by a continuous-time Markov chain $\mathbf{X}:=\{X_t\}_{t\in\mc{T}}$ on $(\Omega,\mc{F},P)$ with finite state space. Without loss of generality we can identify the state space with $\mc{S}:=\{e_1,e_2,\ldots,e_N \}\subset \mathbb{R}^N$, where $e_i\in\mathbb{R}^N$ and the $j$th component of $e_i$ is the Kronecker delta $\delta_{ij}$ for each $i,j=1,\ldots,N$.
	
	We suppose $\mathbf{X}$ and $\mathbf{W}$ are independent processes. In this section, we introduce a model for a price process whose drift and diffusion coefficients depend on the chain $\mathbf{X}$ and in addition has a jump term. 
	
	Suppose that the dynamics of a process $\{\ol{S}_t\}_{t\in\mc{T}}$ are described by
	\begin{align*}
		d\ol{S}_t=\mu_t\ol{S}_tdt+\sigma_t\ol{S}_tdW_t,\quad \ol{S}_0=s.
	\end{align*} 
	Here, $\{\mu_t\}_{t\in\mc{T}}$ and $\{\sigma_t\}_{t\in\mc{T}}$ depend on $\mathbf{X}$ and are given, respectively, by
	$\mu_t:=\mu(t,X_t)=\langle\mu,X_t\rangle$ and $\sigma_t:=\sigma(t,X_t)=\langle\sigma,X_t\rangle$,
	where $\mu=(\mu_1,\mu_2,\ldots,\mu_N)^{\top}\in \mathbb{R}^N$ and $\sigma=(\sigma_1,\sigma_2,\ldots,\sigma_N)^{\top}\in \mathbb{R}^N$ with $\sigma_i>0$ for all $i=1,\ldots, N$. Let $Y_t:=\ln(\ol{S}_t/\ol{S}_0)$ denote the logarithmic return of the process $\{\ol{S}_t\}_{t\in\mc{T}}$. Then, the dynamics of the process $\{\ol{S}_t\}_{t\in\mc{T}}$ can be written as \[\ol{S}_t=\ol{S}_0\exp \ol{Y}_t,\] where $d\ol{Y}_t=\left(\mu_t-\frac{1}{2}\sigma^2_t\right)dt+\sigma_tdW_t.$
	
	Consider now a model for a price process $S_t:=\ol{S}_t\langle\alpha,X_t\rangle$, where $\alpha:=(\alpha_1,\ldots,\alpha_N)^{\top}\in\mathbb{R}^N$ with $\alpha_i>0$ for all $i=1,\ldots,N$. Furthermore, for every $i=1,\ldots,N$, the components $\alpha_i$ satisfy exactly one of the following: (a) $0<\alpha_i<1$, (b) $\alpha_i=1$, or (c) $\alpha_i>1$. The process $S_t$ can be interpreted as a scaled version of the original process. This serves as a pure jump component of the price process.
	
	Let $A:=[a_{ji}]_{i,j=1,\ldots,N}$ be a family of rate matrices of the chain $\mathbf{X}$ under $P$, where $a_{ji}$ is the instantaneous transition intensity of the chain from state $e_i$ to state $e_j$. Note that for $i,j=1,\ldots,N$, the $a_{ji}$ satisfy (1) $a_{ji}\geq 0$, for $i\neq j$ and (2) $\sum_{j=1}^Na_{ji}=0$. The second condition implies that $a_{ii}\leq 0$ for $i=1,\ldots, N$.
	
	Let $\mathbb{F}^{\mathbf{X}}:=\{\mc{F}_t^{\mathbf{X}}\}_{t\in\mc{T}}$ be the $P$-augmented natural filtration generated by $\mathbf{X}$. Then the semimartingale dynamics for $\mathbf{X}$ under $P$ are given by 
	\begin{equation}\label{chain}
		X_t=X_0+\int_{0}^{t} AX_sds+M_t,
	\end{equation}
	where $\mathbf{M}:=\{M_t\}_{t\in \mc{T}}$ is an $\mathbb{R}^N$-valued, square-integrable, $(\mathbb{F}^{\mathbf{X}},P)$-martingale  \cite{elliott:hmm}.
	
	The dynamics of $\{S_t\}_{t\in\mc{T}}$ are then described by
	\begin{align}\label{prior change of measure}
		dS_t&= \langle\alpha,X_t\rangle d\ol{S}_t +\ol{S}_{t-}\langle\alpha,dX_t\rangle\nonumber\\
		&=\mu_tS_tdt+\sigma_tS_tdW_t +\ol{S}_t\langle\alpha,AX_t\rangle dt+\ol{S}_t\langle\alpha,dM_t\rangle.
	\end{align}
	Note that $\ol{S}_{t-}=\ol{S}_t$ since the process $\{\ol{S}_t\}_{t\in\mc{T}}$ is continuous.
	
	\section{Controlled Markov chains}\label{controlsection}
	We review here ideas relating to the control of a Markov chain as discussed in \cite{soner:controlledprocess, miller:control, miller:controlmarkov, miller:control2}. Consider the same complete probability space $(\Omega,\mc{F},P)$. Suppose $\textsf{U}$ is a set of control values where \textsf{U} is a compact subset of $\mathbb{R}^N_+$. Define the set of admissible control functions, denoted by $\mc{U}$, as the set of $\mathbb{F}^{\mathbf{X}}$-predictable functions $u:\mc{T}\rightarrow\textsf{U}$. We wish to consider the situation where the rate matrix $A$ in \eqref{chain} is replaced by a family of rate matrices $B(u):=[b_{ji}(u)]_{i,j=1,\ldots,N}$.
	
	As above, for $i,j=1,\ldots,N$ and $u\in\mc{U}$, the terms $b_{ji}(u)$ satisfy (1) $b_{ji}(u)\geq0$, for $i\neq j$ and (2) $\sum_{j=1}^Nb_{ji}(u)=0$. The second condition then implies that $b_{ii}(u)\leq0$ for $i=1,\ldots, N$.
	
	In the next section, we will show how for any control function $u\in\mc{U}$, the chain $\mathbf{X}=\mathbf{X}^u=\{X^u_t\}_{t\in\mc{T}}$ then has dynamics
	\begin{equation}\label{chainb}
		X^u_t=X^u_0+\int_0^t B(u(s))X^u_sds+M^u_t,
	\end{equation}
	where $\mathbf{M}^u:=\{M^u_t\}_{t\in \mc{T}}$ is an $\mathbb{R}^N$-valued, square-integrable martingale with initial condition $M^u(0)=0$ and (matrix) quadratic variation
	\begin{align}\label{quadchar}
		\langle M^{u}\rangle_t&=\int_0^t\left\{\mbox{diag}[B(u(s))X^u_s]-B(u(s))\mbox{diag}[X^u_s]-\mbox{diag}[X^u_s]B(u(s))^{\top} \right\}ds.
	\end{align}
	The chain $\mathbf{X}^u$ is then a controlled Markov chain. In the context of asset pricing under regime-switching, we can interpret the control $u(t)$ as a factor that influences the rate matrix at time $t$. The probabilities of each state at each time $t$ are then changing depending on $u(t)$.

	\section{Change of Measure for Markov Chains}\label{change of measure for markov chains}
	In this section, we provide a new approach and show how the dynamics \eqref{chainb} can be constructed using a Girsanov-type of measure change using methods from \cite{dufour:filter}. The results will be presented without proofs as the proofs are similar. We begin by considering a control function $u(t)=(u_1(t),\ldots, u_N(t))^{\top}\allowbreak\in\mc{U}$ and, as above, the family of rate matrices $B(u(t)):=[b_{ji}(u_i(t))]_{i,j=1,\ldots,N}$ for the chain $\mathbf{X}$.
	
	Recall that the off-diagonal entries of the rate matrix represent the rate of switching between states. Having large-valued off-diagonal entries implies that jumping to another state takes place faster. Therefore, the chain can model rapid changes in asset prices.
	
	Suppose now $a_{ji}>0$ for all $i,j=1,\ldots,N$. Define, for each $t\in\mc{T}$,
	\begin{align*}
		D(u(t)):=\left[\frac{b_{ji}(u_i(t))}{a_{ji}}\right]_{i,j=1,\ldots,N}=[d_{ji}(u_i(t))]_{i,j},
	\end{align*}
	and write
	\begin{align*}
		D_0(u(t))&:=D(u(t))-\mbox{diag}(d_{11}(u_1),\ldots,d_{NN}(u_N)).
	\end{align*}
	Similarly,  write 
	\begin{align*}
		A_0:=A-\mbox{diag}(a_{11},a_{22},\ldots,a_{NN}).
	\end{align*}
	These matrices can be thought of as the original matrices $D(u(t))$ and $A$, but with zeroes as their diagonal entries.
	
	\begin{mypr}
		Let $\mathbf{J}:=\{J(t)\}_{t\in\mc{T}}$ be a vector-valued counting process on $(\Omega,\mathcal{F},P)$, where for each $t\in\mc{T}$, \[J(t):=\int_0^t(I-\mathbf{diag}[X_{s-}])dX_s.\] Then $J_i(t)$ counts the number of jumps of chain $\mathbf{X}$ to state $e_i$ up to time t, for each $i=1,\ldots,N$.
	\end{mypr}
	\begin{proof}
		Let $0< s\leq t$. Suppose $X_{s-}=e_i$ and $\Delta X_s=X_s-X_{s-}$.  The vector $\Delta X_s$ is the zero vector unless a jump to another state occurs. Then, $(I-\mathbf{diag}[X_{s-}])\Delta X_s=X_s$. Summing over all $s\in(0,t]$ yields the result.
	\end{proof}
	
	\begin{mypr}\label{Jbar}
		The process $\overline{\mathbf{J}}:=\{\overline{J}(t)\}_{t\in\mc{T}}$ defined by 
		\[\overline{J}(t):=J(t)-\int_0^tA_0X_sds,\quad t\in\mc{T}\] is an $\mathbb{R}^N$-valued, $(\mathbb{F}^{\mathbf{X}},P)$-martingale.
	\end{mypr}
	\begin{proof}
		\begin{align*}
			(I-\mbox{diag}[X_{t-}])dX_t	&=(I-\mbox{diag}[X_{t-}])(AX_tdt+dM_t) \\
			dJ(t)&=A_0X_tdt+(I-\mbox{diag}[X_{t-}])dM_t \\
			&=dJ(t)-d\overline{J}(t)+(I-\mbox{diag}[X_{t-}])dM_t.
		\end{align*}
		Rearranging the equation yields
		\begin{align*}
			d\overline{J}(t)	&=(I-\mbox{diag}[X_{t-}])dM_t,
		\end{align*}
		which proves the result.
	\end{proof} 
	Define for each $t\in\mc{T}$ and $u\in\mc{U}$,
	\[\Lambda^u_1(t):=1+\int_0^t\Lambda^u_1(s-)[D_0(u(s-))X_{s-}-\mathbf{1}]^{\top}d\overline{J}(s),\] where $\mathbf{1}:=(1,\ldots,1)^{\top}\in\mathbb{R}^N$. From Proposition \ref{Jbar}, $\{\Lambda^u_1(t)\}_{t\in\mc{T}}$ is an $(\mathbb{F}^{\mathbf{X}},P)$-local martingale. Since $u(t)$ is bounded, $\{\Lambda^u_1(t)\}_{t\in\mc{T}}$ is an $(\mathbb{F}^{\mathbf{X}},P)$-martingale.
	
	The control of the chain $\mathbf{X}$ will be described by the following change of measure using the density $\Lambda_1^u$. For each $u\in\mc{U}$, define a new probability measure $P^{(u)}\sim P$ on $\mc{F}_T$:
	\[\frac{dP^{(u)}}{dP}\bigg|_{\mc{F}_T}:=\Lambda^u_1(T).\]

	\begin{mypr}\label{girsanovchain}
		For each $u\in\mc{U}$, the chain $\mathbf{X}$ is a Markov chain with rate matrix $B(u)$ under $P^{(u)}$.
	\end{mypr}
	\begin{proof}
		Write $B_0(u(t)):=B(u(t))-\mbox{diag}(b_{11}(u_1(t)),\ldots,b_{NN}(u_N(t)))$. Write 
		\begin{align*}
			J^*(t)=J(t)-\int_0^tB_0(u(s))X_sds.
		\end{align*}
		The dynamics of $\{\Lambda_1(t)J^*(t)\}_{t\in\mc{T}}$ are then given by
		\begin{align*}
			\Lambda^u_1(t)J^*(t)
			&= \ds\sum_{0<s\leq t}\Lambda^u_1(s-)\Delta J(s)-\int_0^t\Lambda^u_1(s)B_0(u(s))X_{s}ds+\int_0^tJ^*(s-)d\Lambda^u_1(s) \\
			&\quad+\ds\sum_{0<s\leq t}\Lambda^u_1(s-)[D_0(u(s-))X_{s-}-\mathbf{1}]^{\top}\Delta J(s)\Delta J(s)^{\top}	\\
			&=\ds\ds\sum_{0<s\leq t}\Lambda^u_1(s-)\Delta J(s)-\int_0^t\Lambda^u_1(s)B_0(u(s))X_{s}ds+\int_0^tJ^*(s-)d\Lambda^u_1(s)\\
			&\quad+\ds\int_0^t\Lambda^u_1(s-)\mbox{diag}[\Delta J(s)][D_0(u(s-))X_{s-}-\mathbf{1}]	\\
			&=\ds\ds\sum_{0<s\leq t}\Lambda^u_1(s-)\Delta J(s)-\int_0^t\Lambda^u_1(s)B_0(u(s))X_{s}ds+\int_0^tJ^*(s-)d\Lambda^u_1(s)\\
			&\quad+\ds\int_0^t\Lambda^u_1(s-)\mbox{diag}[A_0X_{s-}][D_0(u(s))X_{s-}-\mathbf{1}]ds	\\
			&\quad+\ds\int_0^t\Lambda^u_1(s-)\mbox{diag}[\Delta\overline{J}(s)][D_0(u(s))X_{s-}-\mathbf{1}].	
		\end{align*}
		Since $\mbox{diag}[A_0X_s][D_0(u(s))X_s-\mathbf{1}]=B_0(u(s))X_s-A_0X_s$, then
		\begin{align*}
			\Lambda^u_1(t)J^*(t)&=\ds\int_0^t\Lambda^u_1(s-)d\overline{J}(s)+\int_0^tJ^*(s-)d\Lambda^u_1(s)\\
			&\quad+\ds\int_0^t\Lambda^u_1(s-)\mbox{diag}[d\overline{J}(s)][D_0(u(s-))X_{s-}-\mathbf{1}].
		\end{align*}
		Therefore, $\{\Lambda^u_1(t)J^*(t)\}_{t\in\mc{T}}$ is an $(\mathbb{F}^{\mathbf{X}},P)$-martingale. This implies that, under $P^{(u)}$, $\{J^*(t)\}_{t\in\mc{T}}$ is an $(\mathbb{F}^{\mathbf{X}},P^{(u)})$-martingale. Then,
		\begin{align*}
			(I-\mbox{diag}[X_{t-}])dX_t&=A_0X_tdt+(I-\mbox{diag}[X_{t-}])dM_t \\
			&=dJ^*(t)-d\overline{J}(t)+B_0(u(t))X_tdt+(I-\mbox{diag}[X_{t-}])dM_t\\
			&=dJ^*(t)+B_0(u(t))X_tdt.
		\end{align*}
		Pre-multiplying by $(I-X_{t-}\mathbf{1}^{\top})$ gives
		\begin{align*}
			(I-X_{t-}\mathbf{1}^{\top})(I-\mbox{diag}[X_{t-}])dX_t&=(I-X_{t-}\mathbf{1}^{\top})dJ^*(t)+(I-X_{t-}\mathbf{1}^{\top})B_0(u(t))X_{t-}dt.
		\end{align*}
		Then,
		\begin{align*}
			dX_t&=(I-X_{t-}\mathbf{1}^{\top})dJ^*(t)+B(u(t))X_tdt.
		\end{align*}
		Hence, the result.
	\end{proof}
	Consequently, under $P^{(u)}$, the chain $\mathbf{X}=\mathbf{X}^u$ has the following semimartingale dynamics
	\begin{align*}
		X^u_t=X^u_0+\int_0^tB(u(s))X^u_sds+M^{u}_t.
	\end{align*}
	Here, $\mathbf{M}^u:=\{M^u_t\}_{t\in\mc{T}}$ is an $\mathbb{R}^N$-valued, $(\mathbb{F}^{\mathbf{X}},P^{(u)})$-martingale with initial condition $M^{u}(0)=0$ and quadratic variation \eqref{quadchar}. 
	
	\section{An Example of a Rate Matrix B(u)}\label{example of rate matrix B}
	
	In this section, we provide an example of a family of rate matrices $B(u(t))$. Suppose $u_i:\mc{T}\times\Omega\rightarrow\textsf{U}_i$ for each $i=1,\ldots,N$. Write $u(t)=(u_1(t),\ldots,u_N(t))^{\top}$ and define
	\[B(u(t)):=A\mathbf{diag}[u(t)].\]
	Then, for all $i,j=1\ldots,N$, with $i\neq j$, $b_{ji}(u_i(t))=u_i(t)a_{ji}\geq 0$, $t\in\mc{T}$. Furthermore,
	\[\sum_{j=1}^{N}b_{ji}(u_i(t))=u_i(t)\sum_{j=1}^{N}a_{ji}=0,\] which implies that $b_{ii}(u_i(t))\leq 0$ for each $i=1,\ldots,N$. This proves that $B(u(t))$ is indeed a family of rate matrices.
	
	Further, for each $t\in\mc{T}$, \[D(u(t)):=\left[\frac{b_{ji}(u_i(t))}{a_{ji}}\right]_{i,j=1,\ldots,N}=\left[{\begin{array}{cccc}
			u_1(t)&u_1(t)&\cdots&u_1(t)\\
			u_2(t)&u_2(t)&\cdots&u_2(t)\\
			\vdots   &\vdots   &\ddots&\vdots\\ 
			u_N(t)&u_N(t)&\cdots&u_N(t)\\
	\end{array}}\right]\]
	and
	\[D_0(u(t)):=D(u(t))-\mbox{diag}[u(t)].\]

	\section{Esscher Transform}\label{esscher}
	
	An earlier change of measure adopted in actuarial science \cite{gerbershiu:esscher} is the Esscher transform. It was first introduced by Esscher in \cite{esscher:esscher}. We extend this to our regime-switching dynamics. Write $\mathbb{F}^{\mathbf{W}}:=\{\mc{F}_t^{\mathbf{W}}\}_{t\in\mc{T}}$ for the $P$-augmentation of the natural filtration generated by $\mathbf{W}$. For each $t\in\mc{T}$, we define $\mc{G}_t:=\mc{F}_t^{\mathbf{X}}\lor \mc{F}_t^{\mathbf{W}}$, the minimal augmented $\sigma$-field generated by the two $\sigma$-fields $\mc{F}_t^{\mathbf{X}}$ and $\mc{F}_t^{\mathbf{W}}$. Write $\mathbb{G}:=\{\mc{G}_t\}_{t\in\mc{T}}$. Let $\theta_t:=\theta(t,X_t)$ be a regime switching Esscher parameter such that 
	\begin{align*}
		\theta_t=\langle\theta(t),X_t\rangle,
	\end{align*}
	where $\theta(t):=(\theta_1(t),\ldots,\theta_N(t))^{\top}\in\mathbb{R}^N$. Then, the regime switching Esscher transform on $\mc{G}_t$ with respect to $\{\theta_s \}_{s\in [0,t]}$ is given by:
	\begin{equation*}
		\Lambda_2(t):=\frac{\exp\left(\int_0^t\theta_s d\ol{Y}_s\right)}{E^P\left[\exp\left(\int_0^t\theta_s d\ol{Y}_s\right)\bigg| \mc{F}^{\mathbf{X}}_T\right]}, \quad t\in\mc{T}.
	\end{equation*}
	Using Ito's formula and taking conditional expectations given $\mc{F}^{\mathbf{X}}_T$ yields
	\begin{align*}
		Z_t&=Z_0+\int_0^t\left[\theta_s\left(\mu_s-\frac{1}{2}\sigma^2_s\right)+\frac{1}{2}\theta^2_s\sigma^2_s\right]Z_sds+\int_0^t\theta_s\sigma_sZ_sdW_s.
	\end{align*}
	Taking the conditional expectation of $Z_t$ given $\mc{F}^{\mathbf{X}}_T$ yields
	\begin{align*}
		E^P\left[Z_t|\mc{F}^{\mathbf{X}}_T\right]&=Z_0+\int_0^t\left[\theta_s\left(\mu_s-\frac{1}{2}\sigma^2_s\right)+\frac{1}{2}\theta^2_s\sigma^2_s\right]E^P[Z_s|\mc{F}^{\mathbf{X}}_T]ds\\
		&=Z_0\exp\left[\ds\int_0^t\theta_s\left(\mu_s-\frac{1}{2}\sigma_s^2\right)ds+\frac{1}{2}\int_0^t\theta^2_s\sigma^2_sds \right] .
	\end{align*}
	Thus, for $t\in\mc{T}$, the Radon-Nikodym derivative of the regime switching Esscher transform is given by
	\begin{align*}
		\Lambda_2(t)&=\frac{\exp\left[\ds\int_0^t\theta_s\left(\mu_s-\frac{1}{2}\sigma_s^2\right)ds+\int_0^t\theta_s\sigma_sdW_s\right]}{\exp\left[\ds\int_0^t \theta_s\left(\mu_s-\frac{1}{2}\sigma^2_s\right)ds+\frac{1}{2}\int_0^t\theta^2_s\sigma^2_sds \right]}\\
		&=\exp\left(\int_0^t\theta_s\sigma_s dW_s-\frac{1}{2}\int_0^t\theta^2_s\sigma^2_sds\right).
	\end{align*}
	Since $\theta_t$ and $\sigma_t$ are bounded, then $E^P\left[\exp\left(\frac{1}{2}\int_0^t|\theta_s\sigma_s|^2ds\right)\right]<\infty$. This implies that $\Lambda_2(t)$ is a $(\mathbb{G},P)$-martingale. Using the density $\Lambda_2$, we can implement a measure change for the continuous part of the price process in \eqref{prior change of measure}.
	
	Consider now the $\mathbb{G}$-adapted process $\Lambda^u:=\{\Lambda^u(t)\}_{t\in\mc{T}}$ defined by
	\begin{align*}
		\Lambda^u(t):=\Lambda^u_1(t)\cdot\Lambda_2(t), \quad t\in\mc{T}.
	\end{align*}
	Define a new probability measure $Q^u\sim P$ on $\mc{G}_T$:
	\begin{align*}
		\frac{dQ^u}{dP}\bigg|_{\mc{G}_T}:=\Lambda^u(T).
	\end{align*}
	
	As $\mathbf{X}$ and $\mathbf{W}$ are independent, using the standard Girsanov's Theorem for Brownian motion, the process $\mathbf{W}^{\theta}:=\{W^{\theta}_t\}_{t\in\mc{T}}$ defined by
	$W^{\theta}_t:=W_t-\int_0^t\theta_u\sigma_udu$ is a $(\mathbb{G},Q^u)$-standard Brownian motion. Furthermore, using Proposition \ref{girsanovchain}, the chain $\mathbf{X}$ remains a Markov chain with rate matrix $B(u(t))$ under $Q^u$.
	
	The dynamics under $Q^u$ of $S_t$ can then be rewritten as 
	\begin{align*}
		dS_t=(\mu_t+\theta_t\sigma_t)S_tdt+\sigma_tS_tdW^{\theta}_t +\ol{S}_t\langle\alpha,B(u(t))X^u_t\rangle dt+\ol{S}_{t}\langle\alpha,dM^{u}_t\rangle.
	\end{align*}
	
	By the fundamental theorem of asset pricing for semimartingales bounded from below, arbitrage opportunities do not exist if and only if there exists an equivalent local martingale measure under which the discounted price process is a martingale \cite{delbaenschachermayer:FTOAP}. Proposition \ref{emm} below gives a condition that allows the discounted price process to be a martingale. To establish this result, we first state the following lemma.
	
	\begin{mylm}\label{Sbar in terms of S}
		Define $\alpha^{-1}:=(\alpha_1^{-1},\alpha_2^{-1},\ldots,\alpha_N^{-1})^{\top}\in\mathbb{R}^N$. Then
		\begin{align*}
			\ol{S}_t\langle\alpha,B(u(t))X^u_t\rangle=S_t\left\langle\alpha,B(u(t))\mathbf{diag}(\alpha^{-1})X^u_t\right\rangle.
		\end{align*}
	\end{mylm}
	\begin{proof}Suppose $X^u_t=e_i$. Then
		\begin{align*}
			\ol{S}_t\langle\alpha,B(u(t))e_i\rangle&=S_t\langle\alpha^{-1},e_i\rangle\langle\alpha,B(u(t))e_i\rangle\\
			&=S_t\left(\alpha_i^{-1}\right)\left(\sum_{j=1}^N\alpha_jb_{ji}\right)\\
			&=S_t\left(\sum_{j=1}^N\alpha_jb_{ji}\alpha_i^{-1}\right)\\
			&=S_t\left\langle\alpha,B(u(t))\mathbf{diag}(\alpha^{-1})e_i\right\rangle,
		\end{align*}
		which proves the result.
	\end{proof}
	Using It\^{o}'s product rule and Lemma \ref{Sbar in terms of S} yields the following proposition.
	\begin{mypr}\label{emm}
		Define a risk-free interest rate by $r_t:=\langle r,X_t\rangle$, where $r=(r_1,\ldots,r_N)^{\top}\in\mathbb{R}^N$. For each $t\in\mc{T}$ and some $u\in\mc{U}$, let the discounted price process $\{\widetilde{S}_t \}_{t\in\mc{T}}$ be defined by
		\begin{align*}
			\widetilde{S}_t=e^{-\int_0^tr_udu}S_t.
		\end{align*}
		For $i=1,\ldots,N$, define
		\begin{align*}
			\theta^u_i(t)=\frac{1}{\sigma_i}\left(r_i-\mu_i-\sum_{j=1}^N\alpha_jb_{ji}(u_i(t))\alpha_i^{-1}\right).
		\end{align*}
		Then $\{\widetilde{S}_t \}_{t\in\mc{T}}$ is a $(\mathbb{G},Q^u)$-martingale if and only if\[\theta_t=\langle\theta^u(t),X_t^u\rangle, \]
		where $\theta^u(t)=(\theta^u_1(t),\ldots,\theta^u_N(t))^{\top}\in\mathbb{R}^N$.
	\end{mypr}
	\begin{proof}
		Using Ito's product rule and Lemma \ref{Sbar in terms of S},
		\begin{align*}
			d\widetilde{S}_t&= e^{-\int_0^tr_udu}\bigg[-r_tS_tdt+(\mu_t+\theta_t\sigma_t)S_tdt+\sigma_tS_tdW^{\theta}_t+\ol{S}_t\langle\alpha,B(u(t))X^u_t\rangle dt+\ol{S}_t\langle\alpha,dM^{u}_t\rangle\bigg] \\		
			&=\bigg[\mu_t+\theta_t\sigma_t-r_t+\left\langle\alpha,B(u(t))\mathbf{diag}(\alpha^{-1})X^u_t\right\rangle\bigg]\widetilde{S}_tdt+\sigma_t\widetilde{S}_tdW^{\theta}_t+\ol{S}_te^{-\int_0^tr_udu}\langle\alpha,dM^{u}_t\rangle.
		\end{align*}
		This is a $(\mathbb{G},Q^u)$-martingale if and only if the finite variation term is indistinguishable from the zero process.
	\end{proof}
	With the appropriate choice of $\theta_t$ such that $\widetilde{S}_t$ is a $(\mathbb{G},Q^u)$-martingale, the dynamics under $Q^u$ of $S_t$ can be written as:
	\begin{equation*}
		dS_t=r_tS_tdt+\sigma_tS_tdW^{\theta}_t +\ol{S}_t\langle\alpha,dM^{u}_t\rangle
	\end{equation*}
	or
	\begin{align}\label{sbar2 chapter 3}
		dS_t&=r_tS_tdt+\sigma_tS_tdW^{\theta}_t+S_{t-}\langle\alpha^{-1},X^u_{t-}\rangle\langle\alpha,dX^u_t\rangle-S_t\left\langle\alpha,B(u(t))\mathbf{diag}(\alpha^{-1})X^u_t\right\rangle dt.
	\end{align}

	\section{Pricing European Call Options via Homotopy}\label{uend}
	In this section, we shall determine the price of a European call option. We first show that the price satisfies a system of partial differential equations (PDEs). We then solve these PDEs using the homotopy analysis method. The discussion follows closely with that of \cite{elliottchansiu:homotopyoptionpricing}. The results will be presented without proofs.
	
	Consider a vanilla European call option with underlying $S$, strike price $K$, and maturity at time $T$. Given $S_t=s$ and $X_t^u=x$, then the call price at time $t$ is given by
	\begin{align*}
		C(t,s,x):=E^{Q^u}_{t,s,x}\left[e^{-\int_0^tr_udu}(S_T-K)^{+}\right],
	\end{align*}
	where $Q^u$ is a risk-neutral measure and $E^{\cdot}_{t,s,x}$ is the conditional expectation given $S_t=s$ and $X_t=x$. Since $S_t:=\ol{S}_t\langle\alpha,X_t^u \rangle$, then
	\begin{align*}
		C(t,s,x)&=E^{Q^u}_{t,s,x}\left[e^{-\int_t^Tr_{u}du}(S_T-K)^{+}\right]\\ &=E^{Q^u}_{t,s,x}\left[e^{-\int_t^Tr_{u}du}\left(\ol{S}_T\langle\alpha,X^u_T\rangle-K\right)^{+}\right]\\
		&=:\ol{C}(t,\ol{s},x)
	\end{align*} 
	
	\begin{mypr}
		For each $i=1,\ldots,N$, write $\ol{C}_i:=\ol{C}_i(t,\ol{s})=\ol{C}(t,\ol{s},e_i)$ and $\mathbf{\ol{C}}:=(\ol{C}_1,\ldots,\ol{C}_N)^{\top}\in\mathbb{R}^N$. Let $\mc{O}:=(0,T)\times(0,\infty)$ be an open set. Suppose that for each $i=1,\ldots,N$, $\ol{C}_i\in {\mc{C}}^{1,2}(\mc{O})$, where ${\mc{C}}^{1,2}(\mc{O})$ is the space of functions $f(t,\ol{s})$ such that $f$ is continuously differentiable in $t$ and twice continuously differentiable in $\ol{s}$. Then, for each $i=1,\ldots,N$ and each $(t,\ol{s})\in\mc{O}$, $\ol{C}_i$ satisfies the following system of PDEs:
		\begin{align}\label{pde for cbar}
			0&=\frac{\partial \ol{C}_i}{\partial t}+r_i\ol{s}\frac{\partial \ol{C}_i}{\partial \ol{s}}+\frac{1}{2}\sigma_i^2\ol{s}^2\frac{\partial^2 \ol{C}_i}{\partial \ol{s}^2}-r_i\ol{C}_i+\langle \mathbf{\ol{C}},B(u(t))e_i\rangle,
		\end{align} 
		with terminal condition $\ol{C}_i(T,\ol{s}):=\ol{C}(T,\ol{s},e_i)=(\ol{s}-K)^+$.
	\end{mypr}
	\begin{proof}
		
		For each $i=1,\ldots,N$, write $\ol{C}_i:=\ol{C}_i(t,\ol{s})=\ol{C}(t,\ol{s},e_i)$ and $\mathbf{\ol{C}}:=(\ol{C}_1,\ldots,\ol{C}_N)^{\top}\in\mathbb{R}^N$. Then for $x\in\mc{S}$, $\ol{C}(t,\ol{s},x)=\langle\mathbf{\ol{C}},x\rangle$. For $i=1,\ldots,N$, by It\^{o}'s formula,
		\begin{align*}
			d\ol{C}_i&=\frac{\partial \ol{C}_i}{\partial t}dt+\frac{\partial \ol{C}_i}{\partial \ol{s}}d\ol{s}+\frac{1}{2}\sigma_t^2\ol{s}^2\frac{\partial^2 \ol{C}_i}{\partial \ol{s}^2}dt\\
			&=\frac{\partial \ol{C}_i}{\partial t}dt+\frac{\partial \ol{C}_i}{\partial \ol{s}}(r_i\ol{s}dt+\sigma_i\ol{s}dW_t^{\theta})+\frac{1}{2}\sigma_i^2\ol{s}^2\frac{\partial^2 \ol{C}_i}{\partial \ol{s}^2}dt\\
			&=\left(\frac{\partial \ol{C}_i}{\partial t}+r_i\ol{s}\frac{\partial \ol{C}_i}{\partial \ol{s}}+\frac{1}{2}\sigma_i^2\ol{s}^2\frac{\partial^2 \ol{C}_i}{\partial \ol{s}^2}\right)dt+\sigma_i\ol{s}\frac{\partial \ol{C}_i}{\partial \ol{s}}dW_t^{\theta}  
		\end{align*}
		Write $d\mathbf{\ol{C}}(t,\ol{s}):=(d\ol{C}_1,\ldots,d\ol{C}_N)^{\top}\in\mathbb{R}^N$. Then,
		\begin{align}\label{dCbar}
			d\ol{C}(t,\ol{s},X^u_t)=\langle \mathbf{\ol{C}},dX_t^u\rangle+\langle d\mathbf{\ol{C}}(t,\ol{s}),X_t^u\rangle.
		\end{align}
		
		Since $\{(\ol{S}_t,X_t^u)\}_{t\in\mc{T}}$ is jointly Markovian with respect to $\mathbb{G}$ under $Q^u$, then for each $t\in\mc{T}$, the discounted price of a European call option given $\ol{S}_t=\ol{s}$ and $X_t^u=x$ is given by
		\begin{align*}
			\widetilde{C}(t,\ol{s},x)&=e^{-\int_0^tr_{u}du}\ol{C}(t,\ol{s},x)\\
			&=e^{-\int_0^tr_{u}du}E^{Q^u}\left[e^{-\int_t^Tr_{u}du}\left(\ol{S}_T\langle\alpha,X^u_T\rangle-K\right)^{+}|\ol{S}_t=\ol{s},X^u_t=x\right]\\
			&=E^{Q^u}\left[e^{-\int_0^Tr_{u}du}\left(\ol{S}_T\langle\alpha,X^u_T\rangle-K\right)^{+}|\mc{G}_t\right].
		\end{align*}
		Hence, $\{\widetilde{C}(t,\ol{s},X^u_t)\}_{t\in\mc{T}}$ is a $(\mathbb{G},Q^u)$-martingale.
		
		For each $i=1,\ldots,N$, write $\widetilde{C}_i:=\widetilde{C}_i(t,\ol{s})=\widetilde{C}(t,\ol{s},e_i)$ and $\mathbf{\widetilde{C}}:=(\widetilde{C}_1,\ldots,\widetilde{C}_N)^{\top}\in\mathbb{R}^N$. Then for $x\in\mc{S}$, $\widetilde{C}(t,\ol{s},x)=\langle\mathbf{\widetilde{C}},x\rangle$.
		From \eqref{dCbar},
		\begin{align*}
			d\widetilde{C}(t,\ol{s},X^u_t)&=e^{-\int_0^tr_{u}du}\bigg[-r_t\langle\mathbf{\ol{C}},X^u_t\rangle dt+d\ol{C}(t,\ol{s},X^u_t)\bigg]\\
			&=e^{-\int_0^tr_{u}du}\bigg[-r_t\langle\mathbf{\ol{C}},X^u_t\rangle dt+\langle \mathbf{\ol{C}},dX_t^u\rangle+\langle d\mathbf{\ol{C}}(t,\ol{s}),X_t^u\rangle\bigg]\\
			&=e^{-\int_0^tr_{u}du}\bigg[-r_t\langle\mathbf{\ol{C}},X^u_t\rangle dt+\langle \mathbf{\ol{C}},B(u(t))X_t^udt+dM_t^u\rangle+\langle d\mathbf{\ol{C}}(t,\ol{s}),X_t^u\rangle\bigg].
		\end{align*}
		Then, for each $i=1,\ldots,N$,
		\begin{align*}
			d\widetilde{C}(t,\ol{s},e_i)&=e^{-\int_0^tr_{u}du}\bigg[-r_i\ol{C}_i dt+\langle \mathbf{\ol{C}},B(u(t))e_i\rangle dt+\langle \mathbf{\ol{C}},dM_t^u\rangle\\
			&\quad+\left(\frac{\partial \ol{C}_i}{\partial t}+r_i\ol{s}\frac{\partial \ol{C}_i}{\partial \ol{s}}+\frac{1}{2}\sigma_i^2\ol{s}^2\frac{\partial^2 \ol{C}_i}{\partial \ol{s}^2}\right)dt+\frac{\partial \ol{C}_i}{\partial \ol{s}}\sigma_i\ol{s}dW_t^{\theta}  \bigg]\\
			&=e^{-\int_0^tr_{u}du}\bigg[\left(\frac{\partial \ol{C}_i}{\partial t}+r_i\ol{s}\frac{\partial \ol{C}_i}{\partial \ol{s}}+\frac{1}{2}\sigma_i^2\ol{s}^2\frac{\partial^2 \ol{C}_i}{\partial \ol{s}^2}-r_i\ol{C}_i+\langle \mathbf{\ol{C}},B(u(t))e_i\rangle\right)dt\\
			&\quad+\langle \mathbf{\ol{C}},dM_t^u\rangle+\sigma_i\ol{s}\frac{\partial \ol{C}_i}{\partial \ol{s}}dW_t^{\theta}  \bigg].
		\end{align*}
		Since $\{\widetilde{C}(t,\ol{s},X^u_t)\}_{t\in\mc{T}}$ is a $(\mathbb{G},Q^u)$-martingale, the finite variation term should be identical to the zero process. The result then follows.
	\end{proof}
	Since $S_t:=\ol{S}_t\langle\alpha,X_t^u \rangle$ and $C(t,s,x)=\ol{C}(t,\ol{s},x)$, then the following corollary holds.
	
	\begin{mycor}
		Write ${C}_i:={C}_i(t,{s})={C}(t,{s},e_i)$ for each $i=1,\ldots,N$ and the vector $\mathbf{{C}}:=({C}_1,\ldots,{C}_N)^{\top}\in\mathbb{R}^N$. Then, for each $i=1,\ldots,N$ ${C}_i$ satisfies the following system of PDEs:
		\begin{align}\label{C as pde}
			0&=\frac{\partial {C}_i}{\partial t}+r_i{s}\frac{\partial {C}_i}{\partial {s}}+\frac{1}{2}\sigma_i^2{s}^2\frac{\partial^2 {C}_i}{\partial {s}^2}-r_i{C}_i+\langle \mathbf{{C}},B(u(t))e_i\rangle,
		\end{align} 
		with terminal condition ${C}_i(T,{s}):={C}(T,{s},e_i)=({s}-K)^+$.
	\end{mycor}
	\begin{proof}
		Rewriting ${\partial \ol{C}_i}/{\partial \ol{s}}$ and ${\partial^2 \ol{C}_i}/{\partial \ol{s}^2}$ yields:
		\begin{align*}
			\frac{\partial \ol{C}_i}{\partial \ol{s}}=\frac{\partial \ol{C}_i}{\partial {s}}\cdot\frac{\partial s}{\partial \ol{s}}=\frac{\partial}{\partial {s}}\ol{C}(t,\ol{s})\cdot\frac{\partial}{\partial \ol{s}}\ol{s}\langle\alpha,x\rangle=\frac{\partial {C}_i}{\partial {s}}\cdot\langle\alpha,x\rangle
		\end{align*}
		and
		\begin{align*}
			\frac{\partial^2 \ol{C}_i}{\partial \ol{s}^2}=\frac{\partial}{\partial \ol{s}}\left[\frac{\partial \ol{C}_i}{\partial \ol{s}}\right]=\frac{\partial}{\partial \ol{s}}\left[\frac{\partial {C}_i}{\partial {s}}\cdot\langle\alpha,x\rangle\right]=\langle\alpha,x\rangle\frac{\partial}{\partial \ol{s}}\left[\frac{\partial {C}_i}{\partial {s}}\right]=\langle\alpha,x\rangle\frac{\partial}{\partial {s}}\left[\frac{\partial {C}_i}{\partial \ol{s}}\right]=\langle\alpha,x\rangle^2\frac{\partial^2{C}_i}{\partial {s}^2}.
		\end{align*}
		Hence, \eqref{pde for cbar} can be expressed as
		\begin{align*}
			\frac{\partial {C}_i}{\partial t}+r_i{s}\langle\alpha^{-1},x\rangle\frac{\partial {C}_i}{\partial {s}}\cdot\langle\alpha,x\rangle+\frac{1}{2}\sigma_i^2{s}^2\langle\alpha^{-1},x\rangle^2\langle\alpha,x\rangle^2\frac{\partial^2{C}_i}{\partial {s}^2}-r_i{C}_i+\langle \mathbf{{C}},B(u(t))e_i\rangle=0,
		\end{align*}
		which yields the result.
	\end{proof}
	We can solve for the above PDEs using the homotopy analysis method (HAM), first introduced in \cite{liao:homotopyanalysismethod}, which expresses the price of a European call option as an infinite series. The discussion below follows closely with that of \cite{elliottchansiu:homotopyoptionpricing, parkkim:homotopyoptionpricing}.
	
	Using the transformation $x=\log(s)$ and writing
	\begin{align*}
		\mc{L}_i:=\frac{\partial}{\partial t}+r_i\left(\frac{\partial}{\partial x}-1\right)+\frac{1}{2}\sigma_i^2\left(\frac{\partial^2}{\partial x^2}-\frac{\partial}{\partial x}\right),
	\end{align*} 
	it can be shown that
	\begin{align*}
		\frac{\partial {C}_i}{\partial s}(t,s)
		=\frac{\partial {C}_i}{\partial {x}}(t,x)\cdot\frac{1}{s}
	\end{align*}
	and
	\begin{align*}
		\frac{\partial^2 {C}_i}{\partial s^2}(t,s)
		=\frac{1}{s^2}\cdot \frac{\partial^2 C_i}{\partial x^2}(t,x)-\frac{1}{s^2}\cdot\frac{\partial C_i}{\partial x}(t,x).
	\end{align*}
	\eqref{C as pde} can be rewritten as
	\begin{align}\label{C as pde in x}
		0=\frac{\partial {C}_i}{\partial t}(t,x)+r_i\frac{\partial {C}_i}{\partial {x}}(t,x)+\frac{1}{2}\sigma_i^2\left[\frac{\partial^2 {C}_i}{\partial {x}^2}(t,x)-\frac{\partial C_i}{\partial x^2}(t,x)\right]-r_i{C}_i(t,x)+\langle \mathbf{{C}}(t,x),B(u(t))e_i\rangle,
	\end{align} 
	with terminal condition ${C}_i(T,x):={C}(T,{x},e_i)=(e^x-K)^+$. Write
	\begin{align*}
		\mc{L}_i:=\frac{\partial}{\partial t}+r_i\left(\frac{\partial}{\partial x}-1\right)+\frac{1}{2}\sigma_i^2\left(\frac{\partial^2}{\partial x^2}-\frac{\partial}{\partial x}\right).
	\end{align*}
	Then, \eqref{C as pde in x} can be simplified in terms of notation as
	\begin{align*}
		\mc{L}_iC_i(t,x)+\langle\mathbf{{C}}(t,x),B(u(t))e_i\rangle=0.
	\end{align*}
	We then have the following result.
	
	\begin{myth}
		For each $i=1,\ldots,N$, $C_i(t,x)$ can be represented as
		\begin{align*}
			C_i(t,x)=\sum_{m=0}^{\infty}\frac{\hat{C}^m_i(t,x)}{m!},
		\end{align*}
		given the assumption that
		\begin{align*}
			\hat{C}_i^0(t,x)=e^xN(d_{1,i})-Ke^{-r_i(T-t)}N(d_{2,i}),
		\end{align*}
		where
		\begin{align*}
			d_{1,i}=\frac{\log\left(\frac{e^x}{K}\right)+\left(r_i+\frac{1}{2}\sigma_i^2\right)(T-t)}{\sigma_i\sqrt{T-t}},
		\end{align*}
		$d_{2,i}=d_{1,i}-\sigma_i\sqrt{T-t}$, and $N(\cdot)$ is the cumulative distribution function of a standard normal random variable. Write $\tau:=T-t$ and $\beta_i:=\frac{2r_i}{\sigma_i^2}$. Then $\hat{C}^m_i(t,x)$, $m=1,2,\ldots$, is recursively given by
		\begin{align*}
			&\hat{C}^m_i(t,x)=\frac{\breve{C}^m_i(t,x)}{\gamma_i(\tau,x)},\\
			&\breve{C}^m_i(t,x)=\int_t^{T}\int_{\mathbb{R}}G^m_i(T-s,\xi)H_i(T-s,x,\xi) d\xi ds,\\
			&G^m_i(t,x)=\gamma_i(T-t,x)\mc{L}_i\hat{C}^m_i(t,x),\\
			&H_i(t,x,\xi)=\frac{1}{\sqrt{2\pi\sigma_i^2t}}\exp\left[\frac{-(x-\xi)^2}{2\sigma_i^2t}\right],\\
			&\gamma_i(t,x)=\exp\left[\frac{1}{2}(\beta_i-1)x+\frac{1}{8}\sigma_i^2(\beta_i+1)^2t\right].
		\end{align*}
	\end{myth}
	\begin{proof}
		We construct a homotopy of \eqref{C as pde in x} by considering an embedding parameter $p\in[0,1]$ and unknown functions $\hat{C}_i(t,x,p)$ which satisfy the following system:
		\begin{align}\label{homotopy of pde}
			(1-p)\mc{L}_i\left[\hat{C}_i(t,x,p)-\hat{C}_i^0(t,x)\right]+p\left[\mc{L}_i\hat{C}_i(t,x,p)+\langle\mathbf{\hat{C}}(t,x,p),B(u(t))e_i\rangle\right]=0,
		\end{align}
		with terminal condition $\hat{C}_i(T,x,p):=\hat{C}(T,{x},p,e_i)=(e^x-K)^+$ and $\hat{C}_i^0(t,x)=\hat{C}_i(t,x,0)$ for $i=1,\ldots,N$. This is also called an initial guess for $C_i(t,x)$. 
		
		We choose $\hat{C}_i^0(t,x)$ as the Black-Scholes-Merton option price. This implies that $\hat{C}_i^0(t,x)$ satisfies
		\begin{align*}
			\begin{cases}
				\mc{L}_i\hat{C}_i^0(t,x)=0\\
				\hat{C}_i^0(0,x)=(e^x-K)^+.
			\end{cases}
		\end{align*}
		The system in \eqref{homotopy of pde} can then be rewritten as
		\begin{align}\label{simplified homotopy of pde}
			\begin{cases}
				\mc{L}_i\hat{C}_i(t,x,p)+p\langle\mathbf{\hat{C}}(t,x,p),B(u(t))e_i\rangle=0\\
				\hat{C}_i(T,x,p)=(e^x-K)^+,
			\end{cases}
		\end{align}
		with the same terminal condition. Note that if $p=1$, we get
		\begin{align}\label{homotopy p=1}
			\begin{cases}
				\mc{L}_i\hat{C}_i(t,x,1)+\langle\mathbf{\hat{C}}(t,x,1),B(u(t))e_i\rangle=0\\
				\hat{C}_i(T,x,1)=(e^x-K)^+.
			\end{cases} 
		\end{align}
		Consider the Taylor's expansion for $\hat{C}_i(t,x,p)$ about $p=0$ given by
		\begin{align}\label{taylor expansion for c hat}
			\hat{C}_i(t,x,p)=\sum_{m=0}^{\infty}p^m\frac{\hat{C}^m_i(t,x)}{m!},
		\end{align}
		where
		\begin{align*}
			\hat{C}^m_i(t,x)=\frac{\partial^m}{\partial p^m}\hat{C}_i(t,x,p)\bigg|_{p=0}.
		\end{align*}
		Using \eqref{C as pde in x}, \eqref{homotopy p=1}, and \eqref{taylor expansion for c hat} yields
		\begin{align*}
			C_i(t,x)=\lim_{p\to 1^-}\hat{C}_i(t,x,p)=\sum_{m=0}^{\infty}\frac{\hat{C}^m_i(t,x)}{m!}.
		\end{align*}
		We want to solve for $\hat{C}^m_i(t,x)$. Substituting \eqref{taylor expansion for c hat} to \eqref{simplified homotopy of pde} yields the following recursive relations for $m=1,2,\ldots$
		\begin{align}\label{recursive relations}
			\begin{cases}
				\mc{L}_i\hat{C}^m_i(t,x)+\langle\mathbf{\hat{C}}^{m-1}(t,x),B(u(t))e_i\rangle=0\\
				\hat{C}^m_i(T,x)=0.
			\end{cases} 
		\end{align}
		Write $\breve{C}^m_i(t,x)=\gamma_i(\tau,x)\hat{C}_i^m(t,x).$ Using \eqref{recursive relations} yields
		\begin{align*}
			\frac{\partial\breve{C}^m_i}{\partial\tau}+\frac{1}{2}\sigma_i^2\frac{\partial^2\breve{C}^m_i}{\partial x^2}&=-\gamma_i(\tau,x)\mc{L}_i\hat{C}^m_i(t,x)\\
			&=\gamma_i(\tau,x)\langle\mathbf{\hat{C}}^{m-1}(T-\tau,x),B(u(T-\tau))e_i\rangle,
		\end{align*}
		with initial condition $\breve{C}^m_i(0,x)=0$. This type of nonhomogeneous diffusion equation has a known solution (see \cite{parkkim:homotopydiffusion} or Chapter 1.3 of \cite{kevorkian:pde}) given by
		\begin{align*}
			\breve{C}^m_i(t,x)=\int_t^{T}\int_{\mathbb{R}}G^m_i(s,\xi)H_i(T-s,x,\xi) d\xi ds,
		\end{align*}
		which proves the result.
	\end{proof}
	
	\section{Bid and Ask Prices}\label{bid-ask price models}	
	Uncertainty in this model is modelled by uncertainty in the dynamics of the chain $\mathbf{X}$. In turn this represents uncertainty in both the jump and diffusion coefficients of the price process. This is motivated by the buying or selling agent supposing that the market is acting in the most adverse way contrary to their interests. Therefore, a model for the ask price is represented by the supremum over possible trajectories of $\mathbf{X}$ and a model for the bid price by the infimum. These are described by maximum and minimum principles. respectively, for related control problems. A dynamic programming principle and a verification theorem are also given below.
	
	As in \cite{cohen:bsde}, define $\Psi_t$ to be the positive semidefinite matrix given by
	\begin{equation}\label{psi}
		\Psi_t:=\mbox{diag}(AX_{t-})-A\mbox{diag}(X_{t-})-\mbox{diag}(X_{t-})A^{\top}
	\end{equation}
	and, for $Y\in\mathbb{R}^N$,
	\begin{equation*}
		\|Y \|^2_{X_{t-}}=Y^{\top}\Psi_tY.
	\end{equation*}
	Then $\|\cdot\|_{X_{t-}}$ is a stochastic seminorm on $\mathbb{R}^N$ satisfying
	\begin{equation*}
		Y_t^{\top}d\langle M\rangle_tY_t=\|Y_t\|_{X_{t-}}^2dt,
	\end{equation*}
	where $\langle M\rangle$ is the predictable quadratic variation of the martingale defined in \eqref{chain}.
	
	Below $(\textsf{U},d)$ will be a Polish metric space. We adapt some methods used in \cite{elliott:stochastic}. Suppose the maps $f,\sigma,\eta:\mc{T}\times\mathbb{R}\times\mc{S}\mapsto\mathbb{R}$ and $g:\mathbb{R}\times\mc{S}\mapsto\mathbb{R}$ satisfy the following assumption:
	
	\begin{myass}\label{assumption:Lipschitz}
		For each $x\in\mc{S}$ and $t\in\mc{T}$, there exists a constant $C_1>0$ such that for all $z_1,z_2\in{\mathbb{R}}$,
		\begin{align*}
			|f(t,z_1,x)-f(t,z_2,x)|&\leq C_1|z_1-z_2|\\
			|\sigma(t,z_1,x)-\sigma(t,z_2,x)|&\leq C_1|z_1-z_2|\\
			|\eta(t,z_1,x)-\eta(t,z_2,x)|&\leq C_1|z_1-z_2|\\
			|g(z_1,x)-g(z_2,x)|&\leq C_1|z_1-z_2|.
		\end{align*} 
	\end{myass}

	Consider a process $\{Z_t\}_{t\in\mc{T}}$ given by the following dynamics:
	\begin{align}\label{stateequation}
		dZ_t&=f(t,Z_t,X_t)dt+\sigma(t,Z_t,X_t)dW_t+\eta(t,Z_{t-},X_{t-})\langle\alpha,dX_t\rangle.
	\end{align}
	
	\begin{mylm}\label{unique}
		Suppose Assumption \ref{assumption:Lipschitz} holds. Then for any $(t,z,x)\in[0,T)\times\mathbb{R}\times\mc{S}$, the SDE \eqref{stateequation} admits a unique solution with initial conditions $Z_t=z$ and $X_t=x$.
	\end{mylm}
	\begin{proof}
		See Theorem 16.3.11 of \cite{elliott:stochastic} or Theorem 6 of \cite{protter:stochastic} for the proof.
	\end{proof}

	Now, write $S_t=Z_t$. From \eqref{prior change of measure}, the dynamics of $\{Z_t\}_{t\in\mc{T}}$ are given by
	\begin{equation}\label{zdynamics}
		dZ_t=\mu_tZ_tdt+\sigma_tZ_tdW_t +Z_{t-}\langle\alpha^{-1},X_{t-}\rangle\langle\alpha,dX_t\rangle.
	\end{equation}
	Note that \eqref{zdynamics} is of the form \eqref{stateequation} with $f(t,Z_t,X_t)=\mu_tZ_t$, $\sigma(t,Z_t,X_t)=\sigma_tZ_t$ and $\eta(t,Z_{t-},X_{t-})=\langle\alpha^{-1},X_{t-}\rangle Z_{t-}$. The following lemma is a consequence of Lemma \ref{unique}.
	\begin{mylm}
		For any $(t,z,x)\in[0,T)\times\mathbb{R}\times\mc{S}$, the SDE \eqref{zdynamics} admits a unique solution with initial conditions $Z_t=z$ and $X_t=x$.
	\end{mylm}
	\begin{proof}
		For each $x\in\mc{S}$ and $t\in\mc{T}$, choose $C_1=\ds\max\{|\langle\mu,x\rangle|,|\langle\sigma,x\rangle|,|\langle\alpha^{-1},x\rangle|\}$. Then, Assumption \ref{assumption:Lipschitz} holds. By Lemma \ref{unique}, the result follows.
	\end{proof}
	We denote the solution of \eqref{zdynamics} with initial conditions $z$ and $x$ by $\{Z_s^{t,z,x}\}_{s\in[t,T]}$.
	
	Consider an objective functional given by
	\begin{equation}\label{objfunc}
		\mc{J}(t,z,x;u)=\expoc\left[g(Z_T^u,X_T^u)\right].
	\end{equation}
	Here, $\expoc(\cdot)$ is the conditional expectation given $Z^u_{t}=z$ and $X^u_{t}=x$ under $Q^u$. By Assumption \ref{assumption:Lipschitz}, the objective functional \eqref{objfunc} is well-defined. Write $E^{Q^u}[g(\z_t,\x_t)]\allowbreak=\expoc\left[g(Z_t^u,X_t^u)\right]$. 
	
	We shall consider predictable processes $\varphi_1:\Omega\times\mc{T}\mapsto\mathbb{R}$ and $\varphi_2:\Omega\times\mc{T}\mapsto\mathbb{R}^N$ which satisfy the following assumption.
	\begin{myass}\label{square integrable}
		For each $x\in\mc{S}$, $t\in\mc{T}$ and $u\in\textsf{U}$,
		\begin{align*}
			E\bigg[\int_0^T(\|\mc{J}(t,z,x,u)\|^2+\|\varphi_1(t)\|^2+\|\varphi_2(t)\|_{X_{t-}}^2)dt\bigg]<\infty.
		\end{align*}		
	\end{myass}
	For the rest of the section, we assume that Assumption \ref{assumption:Lipschitz} and Assumption \ref{square integrable} hold.
	
	We shall show that \eqref{objfunc} is a solution to a backward stochastic differential equation (BSDE). We first prove the following lemmas.
	
	Recall from Section \ref{change of measure for markov chains}  that $D(u(t))=[b_{ji}(u_i)/a_{ji}]_{i,j=1,\ldots,N}$ and
	\begin{align*}
		D_0(u)=D(u)-\mbox{diag}(d_{11}(u_1),\ldots,d_{NN}(u_N)).
	\end{align*}
	Furthermore, from Section \ref{esscher}, we have a regime switching parameter $\theta_t=\langle\theta,X_t^u\rangle$, where $\theta:=(\theta_1,\ldots,\theta_N)^{\top}$. The following lemma can be proved by using the fact that $\theta$ and $\sigma$ are bounded by some constant $K$ and expanding the second term via the stochastic seminorm definition.
	
	\begin{mylm}\label{boundedness}
		For each $x\in\mc{S}$, $t\in\mc{T}$ and $u\in\textsf{U}$, there exists a constant $C_2>0$ such that
		\begin{align*}
			|\langle\theta\sigma,x\rangle|^2+\|D_0(u)X_{t-}-\mathbf{1}\|^2_{X_{t-}}<C_1,
		\end{align*}
		where 
		\begin{align*}
			C_1&=K+\max_{u\in\textsf{U}}\left\{\sum_{\substack{k=1\\k\neq j}}^N\frac{(b_{kj}(u_j))^2}{a_{kj}}\right\}.
		\end{align*}
	\end{mylm}
	\begin{proof}
		Since $\theta$ and $\sigma$ are bounded, then the first term is bounded by some constant $K$. Suppose $X_{t-}=e_j\in\mc{S}$. For $i,m,n=1,\ldots,N$ and some vector $\mathbf{a}=(a_1,\ldots,a_N)^{\top}$, define the $N\times N$ matrices $\mbox{col}_i(\mathbf{a})=[c_{i,mn}]$ and $\mbox{row}_i(\mathbf{a})=[r_{i,mn}]$ such that 
		\begin{equation*}
			c_{i,mn}=r_{i,nm}=\begin{cases}
				a_m,&\mbox{if $n=i$}\\
				0,&\mbox{otherwise}.
			\end{cases}
		\end{equation*}
		Then,
		\begin{align*}
			D_0(u)X_{t-}-\mathbf{1}&=(d_{1j}-1,\ldots,-1,\ldots,d_{Nj}-1)^{\top}\\
			\mbox{diag}(AX_{t-})&=\mbox{diag}(a_{1j},\ldots,a_{Nj})\\
			A\mbox{diag}(X_{t-})&=\mbox{col}_j(a_{1j},\ldots,a_{Nj})\\
			\mbox{diag}(X_{t-})A^{\top}&=\mbox{row}_j(a_{1j},\ldots,a_{Nj}).
		\end{align*}
		From \eqref{psi}, the diagonal entries, the $j$th column and the $j$th row of $\Psi_t$ have entries of the form $\pm a_{mj}$ and zero elsewhere. That is,
		\begin{align*}
			\Psi_t=\left[{\begin{array}{ccccccc}
					a_{1j}&0&\cdots&-a_{1j}&\cdots&0\\
					0&a_{2j}&\cdots&-a_{2j}&\cdots&0\\
					\vdots   &\vdots   &\ddots&\vdots&\vdots&\vdots\\ 
					-a_{1j}&-a_{2j}&\cdots&-a_{jj}&\cdots&-a_{Nj}\\
					\vdots&\vdots&\vdots&\vdots&\vdots&\vdots\\
					0&0&\cdots&-a_{Nj}&\cdots&a_{Nj}
			\end{array}}\right].
		\end{align*}
		Multiplying $(D_0(u)X_{t-}-\mathbf{1})$ to the right of $\Psi_t$ yields
		\begin{align*}
			\Psi_t(D_0(u)X_{t-}-\mathbf{1})=\left[{\begin{array}{c}
					d_{1j}(u)a_{1j}-a_{1j}+a_{1j}\\
					d_{2j}(u)a_{2j}-a_{2j}+a_{2j}\\
					\vdots\\
					-\ds\sum_{\substack{k=1\\k\neq j}}^Nd_{kj}(u)a_{kj}+\ds\sum_{k=1}^Na_{kj}\\
					\vdots\\
					d_{Nj}(u)a_{Nj}-a_{Nj}+a_{Nj}
			\end{array}}\right].
		\end{align*}
		From Sections 3 and 4, we have $\sum_{k=1}^Na_{kj}=0$ and $d_{kj}(u)a_{kj}=b_{kj}(u_j)$. Then, multiplying $(D_0(u)X_{t-}-\mathbf{1})^{\top}$ to the left of the above vector yields
		\begin{align*}
			\|D_0(u)X_{t-}-\mathbf{1}\|^2_{X_{t-}}&=\ds\sum_{\substack{k=1\\k\neq j}}^N[b_{kj}(u_j)d_{kj}(u)-b_{kj}(u_j)]+\ds\sum_{\substack{k=1\\k\neq j}}^Nb_{kj}(u_j)\\
			&=\ds\sum_{\substack{k=1\\k\neq j}}^N\frac{(b_{kj}(u_j))^2}{a_{kj}}.
		\end{align*}
		Hence, $\|D_0(u)X_{t-}-\mathbf{1}\|^2_{X_{t-}}$ is also bounded. Choose 
		\begin{align*}
			C_2&=K+\max_{u\in\textsf{U}}\left\{\sum_{\substack{k=1\\k\neq j}}^N\frac{(b_{kj}(u_j))^2}{a_{kj}}\right\}.
		\end{align*}
		The result then follows.
	\end{proof}
	The following lemma is a result of using the predicatable quadratic variation of $\mc{E}(\Theta)$, Gr\"{o}nwall's inequality, and Lemma \ref{boundedness}.
	\begin{mylm}\label{squareintegrablemartingale}
		Define
		\begin{equation*}
			\Theta_t=\int_0^t\langle\theta\sigma,X_s\rangle dW_s+\int_0^t\langle D_0(u)X_{s-}-\mathbf{1}, dM_s\rangle.
		\end{equation*}
		Then the stochastic exponential $\mc{E}(\Theta)$ is a square integrable martingale for $t\leq T$.
	\end{mylm}
	\begin{proof}
		The predictable quadratic variation of $\mc{E}(\Theta)$ is given by
		\begin{equation*}
			\langle\mc{E}(\Theta)\rangle_t=\int_0^t\mc{E}(\Theta)^2_{s-}\left(|\langle\theta\sigma,X_s\rangle|^2+\|D_0(u)X_{s-}-\mathbf{1}\|^2_{X_{s-}}\right)ds.
		\end{equation*}
		Since $\mc{E}(\Theta)_{0}=1$, then using It\^{o}'s formula yields 
		\begin{align*}
			\mc{E}(\Theta)_{t}^2=1+2\int_0^t\mc{E}(\Theta)_{t}d\mc{E}(\Theta)_{t}+\langle\mc{E}(\Theta)\rangle_t.
		\end{align*}
		Hence, for some localizing sequence $T_n\uparrow\infty$ such that $\mc{E}(\Theta)_{T_n}$ and $\langle\mc{E}(\Theta)\rangle_{T_n}$ are bounded martingales,
		\begin{align*}
			E\left[1_{\{t\leq T_n\}}\mc{E}(\Theta)^2_t\right]&\leq E\left[\mc{E}(\Theta)_{t\wedge T_n}^2\right]\\
			&=E\left[1+\langle\mc{E}(\Theta)\rangle_{t\wedge T_n}\right]\\
			&=1+\int_0^t E\bigg[1_{\{s\leq T_n\}}\mc{E}(\Theta)^2_{s-}\left(|\langle\theta\sigma,X_s\rangle|^2+\|D_0(u)X_{s-}-\mathbf{1}\|^2_{X_{s-}}\right)\bigg]ds.
		\end{align*}
		By Gr\"onwall's Inequality and Lemma \ref{boundedness},
		\begin{align*}
			E\left[1_{\{t\leq T_n\}}\mc{E}(\Theta_t)^2\right]&\leq\exp\left[\int_0^t\left(|\langle\theta\sigma,X_s\rangle|^2+\|D_0(u)X_{s-}-\mathbf{1}\|^2_{X_{s-}}\right)ds\right]\\
			&\leq\exp\left[\int_0^tC_2\phantom{|}ds\right]=e^{C_2t}.
		\end{align*}
		By monotone convergence, it follows that
		\begin{align*}
			E\left[\mc{E}(\Theta_t)^2\right]\leq e^{C_2T}<\infty,
		\end{align*}
		which proves the result.
	\end{proof}
	\begin{mydef}\label{balanced driver}
		Consider a driver $F:\mc{T}\times\mc{S}\times\textsf{U}\times\mathbb{R}\times\mathbb{R}^N\mapsto\mathbb{R}$ and the maps $\varphi_1:\Omega\times\mc{T}\mapsto\mathbb{R}$ and $\varphi_2,\varphi_2':\Omega\times\mc{T}\mapsto\mathbb{R}^N$.  Suppose there exists a map $\beta:\mc{T}\times\mc{S}\times\textsf{U}\times\mathbb{R}\times\mathbb{R}^N\times\mathbb{R}^N\mapsto\mathbb{R}^N$ such that
		\begin{itemize}
			\item $\beta$ is predictable in $(t,x,u)\in\mc{T}\times\mc{S}\times\textsf{U}$ and Borel measurable in $(\varphi_1,\varphi_2,\varphi_2')\in\mathbb{R}\times\mathbb{R}^N\times\mathbb{R}^N$;
			\item $\beta>-\mathbf{1}$ for all $(\varphi_1,\varphi_2,\varphi_2')$ and $dP\times dt$ almost all $t$; and
			\item for $dP\times dt$ almost all $t$, and for all $(\varphi_1,\varphi_2,\varphi_2')$,
			\begin{align*}
				\ds\sum_{m=1}^N(\varphi^{(m)}_2-\varphi'^{(m)}_2)\beta(t,x,u,\varphi_1,\varphi_2,\varphi_2')^{(m)}\Psi_t^{(m)}=F(t,x,u,\varphi_1,\varphi_2)-F(t,x,u,\varphi_1,\varphi_2'),
			\end{align*}
		\end{itemize}
		where $\varphi^{(m)}_2,\varphi'^{(m)}_2,\beta(t,x,u,\varphi_1,\varphi_2,\varphi_2')^{(m)}$ and $\Psi_t^{(m)}$ are the $m$th element of their respective vectors. Then, $F$ is \textbf{balanced}.
	\end{mydef}
	
	Let $M^2(\mc{T};\mathbb{R}^{N},\mc{G}_t)$ be the set of $\mathbb{R}^{K\times N}$-valued $\mc{G}_t$-adapted square-integrable processes over $\Omega\times\mc{T}$, and $P^2(\mc{T};\mathbb{R}^{N},\mc{G}_t)$ be the set of $\mathbb{R}^N$-valued $\mc{G}_t$-predictable processes $\{\varphi_t\}_{t\in\mc{T}}$ such that $E(\int_0^t\lVert\varphi_s\rVert_{X_{s-}}^2ds)<\infty$. Write
	\begin{align*}
		F(t,x,u,\varphi_1,\varphi_2)&=\varphi_1\langle\theta\sigma,x\rangle+\sum_{m=1}^N\varphi^{(m)}_2[D_0(u)x-\mathbf{1}]^{(m)}\Psi_t^{(m)}.
	\end{align*} We know that $F$ is Lipschitz in $\varphi_1$ and $\varphi_2$ since their coefficients are bounded. In addition, $F$ is balanced since $\beta=D_0(u)\x_{t-}-\mathbf{1}$ satisfies the conditions in Definition \ref{balanced driver}. We then have the following proposition. 
	
	\begin{mypr}\label{J as solution to BSDE}
		For a given control $u\in\mc{U}$, the process
		\begin{equation*}
			\mc{J}(t,z,x,u)=\expoc\left[g(Z_T^u,X_T^u)\right]
		\end{equation*}
		is the unique solution to the BSDE
		\begin{equation*}
			\begin{cases}
				d\mc{J}(t,z,x,u)&=-F\left(t,\x_t,u(t),\varphi_1(t),\varphi_2(t)\right)dt+\varphi_1(t)dW_t+\langle\varphi_2(t-),dM_t\rangle,\\
				\mc{J}(T,z,x,u)&=g(\z_T,\x_T),
			\end{cases}
		\end{equation*}
		such that $(\varphi_1,\varphi_2)\in M^2(\mc{T};\mathbb{R}^{K\times N},\mc{G}_t)\times P^2(\mc{T};\mathbb{R}^{K\times N},\mc{G}_t)$.	
	\end{mypr}
	\begin{proof}
		We know that $F$ is Lipschitz in $\varphi_1$ and $\varphi_2$ since their coefficients are bounded. In addition, $F$ is balanced since $\beta=D_0(u)\x_{t-}-\mathbf{1}$ satisfies the conditions in Definition \ref{balanced driver}.
		
		By the product rule,
		\begin{align*}
			d\left(\mc{J}(t,z,x,u)\mc{E}(\Theta)_{t}\right)&=\mc{J}(t-,z,x,u)d\mc{E}(\Theta)_{t}+\mc{E}(\Theta)_{t-}d\mc{J}(t,z,x,u)+d[\mc{J},\mc{E}(\Theta)]_t
		\end{align*}
		Since $d\langle M\rangle_t=\Psi_tdt$, then 
		\begin{align*}
			\frac{d\left(\mc{J}(t,z,x,u)\mc{E}(\Theta)_{t}\right)}{\mc{E}(\Theta)_{t-}}
			&=\mc{J}(t-,z,x,u)d\Theta_t+d\mc{J}(t,z,x,u)+d[\mc{J},\Theta]_t\\
			&=\mc{J}(t-,z,x,u)\left[\langle\theta\sigma,x\rangle dW_t+\langle D_0(u)x-\mathbf{1}, dM_t\rangle\right]\\
			&\quad-\bigg[\varphi_1(t)\langle\theta\sigma,x\rangle+\sum_{m=1}^N\varphi^{(m)}_2(t)[D_0(u)x-\mathbf{1}]^{(m)}\Psi_t^{(m)}\bigg]dt\\
			&\quad+\varphi_1(t)dW_t+\langle\varphi_2(t-),dM_t\rangle+\varphi_1(t)\langle\theta\sigma,x\rangle dt\\
			&\quad+\sum_{m=1}^N\varphi^{(m)}_2(t)[D_0(u)x-\mathbf{1}]^{(m)}\Psi_t^{(m)}dt\\
			&=\left[\varphi_1(t)+\mc{J}(t,z,x,u)\langle\theta\sigma,x\rangle\right]dW_t\\
			&\quad+\langle\varphi_2(t-)+ \mc{J}(t-,z,x,u)(D_0(u)x-\mathbf{1}),dM_t\rangle.	
		\end{align*}
		
		The right-hand side of the last equality is a local martingale because the standard Brownian motion $\{W_t\}_{t\in\mc{T}}$ and $\{M_t\}_{t\in\mc{T}}$ are martingales by definition. Therefore, $\mc{J}(t,z,x,u)\mc{E}(\Theta)_{t}$ is also a local martingale.
		
		Write
		\begin{align*}
			\ol{\mc{J}}&=\left(\int_0^T\|\mc{J}(t,Z_t,X_t,u(t))\langle\theta\sigma,X_{t}\rangle\|^2dt\right)^{1/2}.
		\end{align*}
		Then,
		\begin{align*}
			\ol{\mc{J}}&\leq N\left(\max_{X_t\in\mc{S}}\sup_{t\in\mc{T}}|\mc{J}(t,Z_t,X_t,u(t))|^2\right)^{1/2}\left(\int_0^T\|\langle\theta\sigma,X_t\rangle\|^2dt\right)^{1/2}\\
			&\leq\frac{N}{2}\left(\max_{X_t\in\mc{S}}\sup_{t\in\mc{T}}|\mc{J}(t,Z_t,X_t,u(t))|^2+\int_0^T\|\langle\theta\sigma,X_t\rangle\|^2dt\right)<\infty.
		\end{align*}
		Write $\mc{H}^1$ for the set of integrable martingales. By the Burkholder-Davis-Gundy (BDG) Inequality,
		$$\left\{\int_0^s(\varphi_1(s)+\mc{J}(s,Z_s,X_s,u(s))\langle\theta\sigma,X_s\rangle)dW_s\right\}_{s\in(0,T]}\in\mc{H}^1.$$
		Using a similar argument,
		\begin{align*}
			\bigg\{\int_0^s\langle\varphi_2(s-)+ \mc{J}(s-,Z_{s-},X_{s-},u(s-))(D_0(u)X_{s-}-\mathbf{1}),dM_s\rangle\bigg\}_{s\in(0,T]}\in\mc{H}^1.
		\end{align*} 
		By Lemma \ref{squareintegrablemartingale}, we know that $\mc{E}(\Theta)_t$ is a square integrable martingale. Thus, 
		\begin{align*}
			\{\mc{J}(t,Z_t,X_t,u(t))\mc{E}(\Theta)_t\}_{t\geq 0}\in\mc{H}^1.
		\end{align*}
		It then follows that
		\begin{align*}
			E^P\left[\mc{E}(\Theta)_Tg(\z_T,\x_T)|\mc{F}_t\right]&=E^P\left[\mc{E}(\Theta)_T\mc{J}(T,z,x,u)|\mc{F}_t\right]\\
			&=\mc{E}(\Theta)_t\mc{J}(t,z,x,u).
		\end{align*}
		By Bayes' Rule,
		\begin{align*}
			\mc{J}(t,z,x,u)&=\frac{1}{\mc{E}(\Theta)_t}E^P[\mc{E}(\Theta)_Tg(\z_T,\x_T)|\mc{F}_t]\\
			&=E^{Q^{u}}[g(\z_T,\x_T)].
		\end{align*}
	\end{proof}
	Define the following value processes
	\begin{equation}\label{valinf}
		\ul{V}(t,z,x):=\essinf_{u\in\mc{U}}\mc{J}(t,z,x,u)
	\end{equation}
	and
	\begin{equation}\label{valsup}
		\ol{V}(t,z,x):=\esssup_{u\in\mc{U}}\mc{J}(t,z,x,u).
	\end{equation}
	
	The next objective is to show that the value processes \eqref{valinf} and \eqref{valsup} have c\`adl\`ag modifications which are solutions to some BSDEs. We first state the following lemma.
	\begin{mylm}\label{existence of inf and sup}
		Recall that $F$ is a standard balanced driver in Proposition \ref{J as solution to BSDE}. Then, for each $t\in\mc{T}$, $x\in\mc{S}$, $\varphi_1\in\mathbb{R}$, $\varphi_2\in\mathbb{R}^N$, and $u\in\textsf{U}$,
		\begin{enumerate}
			\item[(i)] the maps $(x,\varphi_1,\varphi_2)\mapsto F(t,x,u,\varphi_1,\varphi_2)$ have a common uniform Lipschitz constant K;
			\item[(ii)] $\essinf_{u\in\textsf{U}}\left[\frac{b_{ji}(u)}{a_{ji}}-1\right]> -1$ for $i\neq j$, where $D=\left[\frac{b_{ji}(u)}{a_{ji}}\right]$ as defined in Section \ref{change of measure for markov chains};
			\item[(iii)] $\sup_{u\in\textsf{U}}\{|F(t,x,u,0,0)|^2\}$ is bounded by a predictable $dt\times dP$-integrable process;
			\item[(iv)] the maps $u\mapsto\frac{b_{ji}(u)}{a_{ji}}-1$ are continuous, for fixed $(t,x,\varphi_1,\varphi_2)$, and $\textsf{U}$ is a countable union of compact metrizable subsets of itself.
		\end{enumerate}
		Furthermore, there is a version of the mappings
		\begin{align*}
			\ul{F}(t,x,\varphi_1,\varphi_2)&=\essinf_{u\in\textsf{U}}F(t,x,u,\varphi_1,\varphi_2)\\
			\ol{F}(t,x,\varphi_1,\varphi_2)&=\esssup_{u\in\textsf{U}}F(t,x,u,\varphi_1,\varphi_2)
		\end{align*}
		which are standard balanced BSDE drivers.
	\end{mylm}
	\begin{proof}
		See Lemma 19.3.8 of \cite{elliott:stochastic} for the proof.
	\end{proof}
	
	Using the previous lemma, we have the following result.
	\begin{mylm}\label{existence of H}
		Define the functions
		\begin{equation*}
			\ul{H}(t,z,x,\varphi_1,\varphi_2)=\essinf_{u\in\mc{U}}F(t,x,u,\varphi_1,\varphi_2)
		\end{equation*}
		and
		\begin{equation*}
			\ol{H}(t,z,x,\varphi_1,\varphi_2)=\esssup_{u\in\mc{U}}F(t,x,u,\varphi_1,\varphi_2).
		\end{equation*}
		Then there are versions of $\ul{H}$ and $\ol{H}$ which are balanced Lipschitz drivers for some BSDEs.
	\end{mylm}
	\begin{proof}
		This is a direct application of Lemma \ref{existence of inf and sup}.
	\end{proof}
	
	\begin{mypr}[Dynamic Programming Principle]\label{V has version}
		The value processes $\ul{V}$ and $\ol{V}$ have c\`adl\`ag modifications, which are the respective solutions to the BSDEs
		\begin{equation}\label{V inf}
			\begin{cases}
				d\ul{V}(t,z,x)&=-\ul{H}(t,\x_t,\varphi_1(t),\varphi_2(t))dt+\varphi_1(t)dW_t+\langle\varphi_2(t-),dM_t\rangle\\
				\ul{V}(T,z,x)&=g(\z_T,\x_T)
			\end{cases}
		\end{equation}
		and
		\begin{equation}\label{V sup}
			\begin{cases}
				d\ol{V}(t,z,x)&=-\ol{H}(t,\x_t,\varphi_1(t),\varphi_2(t))dt+\varphi_1(t)dW_t+\langle\varphi_2(t-),dM_t\rangle\\
				\ol{V}(T,z,x)&=g(\z_T,\x_T).
			\end{cases}
		\end{equation}
	\end{mypr}
	\begin{proof}
		As discussed in \cite{pardoux:bsde} for Brownian motions and \cite{cohen:bsde} for Markov chains, BSDEs with drivers $\ul{H}$ and $\ol{H}$ defined in Lemma \ref{existence of H} have c\`adl\`ag solutions, which we denote by $\ul{Y}$ and $\ol{Y}$, respectively. 
		
		By definition, for all $u\in\mc{U}$
		\begin{equation*}
			\ul{H}(t,x,\varphi_1,\varphi_2)\leq F(t,x,u,\varphi_1,\varphi_2)
		\end{equation*}
		and
		\begin{equation*}
			\ol{H}(t,x,\varphi_1,\varphi_2)\geq F(t,x,u,\varphi_1,\varphi_2).
		\end{equation*}
		By the comparison theorem for BSDEs with Brownian motion and Markov chains \cite{cohen:bsde,peng:comparison,cohen:generalcomparison} and Proposition \ref{J as solution to BSDE}, up to indistinguishability,
		\begin{equation*}
			\ul{Y}(t,z,x)\leq \mc{J}(t,z,x,u)
		\end{equation*}
		and
		\begin{equation*}
			\ol{Y}(t,z,x)\geq \mc{J}(t,z,x,u),
		\end{equation*}
		for all $u\in\mc{U}$. By Filippov's Implicit Function Theorem (see Theorem 21.3.4 of \cite{elliott:stochastic}), for every $\epsilon>0$, there exist predictable controls $\ul{u}^{\epsilon},\ol{u}^{\epsilon}\in\mc{U}$ such that
		\begin{equation*}
			F(t,x,\ul{u}^{\epsilon},\varphi_1,\varphi_2)\leq \ul{H}(t,x,\varphi_1,\varphi_2)+\epsilon
		\end{equation*}
		and
		\begin{equation*}
			F(t,x,\ol{u}^{\epsilon},\varphi_1,\varphi_2)\geq \ol{H}(t,x,\varphi_1,\varphi_2)-\epsilon.
		\end{equation*}
		
		Since $\ul{Y}(t,z,x)+\epsilon(T-t)$ and $\ol{Y}(t,z,x)-\epsilon(T-t)$ are the respective solutions to BSDEs with drivers $\ul{H}(t,x,\varphi_1,\varphi_2)+\epsilon$ and $\ol{H}(t,x,\varphi_1,\varphi_2)-\epsilon$, then up to indistinguishability, by the comparison theorem,
		\begin{equation*}
			\mc{J}(t,z,x,\ul{u}^{\epsilon})\leq\ul{Y}(t,z,x)+\epsilon(T-t)
		\end{equation*}
		and
		\begin{equation*}
			\mc{J}(t,z,x,\ol{u}^{\epsilon})\geq\ol{Y}(t,z,x)-\epsilon(T-t).
		\end{equation*}
		Letting $\epsilon\rightarrow 0$, then for every $t\in\mc{T}$
		\begin{equation*}
			\ul{Y}(t,z,x)=\essinf_{u\in\mc{U}}J(t,z,x,u)=\ul{V}(t,z,x)
		\end{equation*}
		and
		\begin{equation*}
			\ol{Y}(t,z,x)=\esssup_{u\in\mc{U}}J(t,z,x,u)=\ol{V}(t,z,x).
		\end{equation*}
		Thus, $\ul{Y}$ and $\ol{Y}$ are versions of $\ul{V}$ and $\ol{V}$, respectively.
	\end{proof}
	The following proposition states the minimum and maximum principles for the control problems.
	\begin{mypr}[Minimum/Maximum Principles]\label{minmax principle}
		Let $\big(\ul{V},\ul{\varphi}_1,\ul{\varphi}_2\big)$ and $\big(\ol{V},\ol{\varphi}_1,\allowbreak\ol{\varphi}_2\big)$ be the resepective solutions to the BSDEs with drivers $\ul{H}$ and $\ol{H}$ and terminal value $g(\z_T,\x_T)$. The controls $\ul{u},\ol{u}\in\mc{U}$ are optimal if and only if
		\begin{equation*}
			F\big(t,\x_t,\ul{u},\ul{\varphi}_1,\ul{\varphi}_2\big)=\ul{H}\big(t,\x_t,\ul{\varphi}_1,\ul{\varphi}_2\big)
		\end{equation*} 
		and
		\begin{equation*}
			F\big(t,\x_t,\ol{u},\ol{\varphi}_1,\ol{\varphi}_2\big)=\ol{H}\big(t,\x_t,\ol{\varphi}_1,\ol{\varphi}_2\big).
		\end{equation*} 
	\end{mypr}
	\begin{proof}
		By definition, $\mc{J}(t,z,x,u)\geq\ul{V}(t,z,x)$ and $\mc{J}(t,z,x,u)\leq\ol{V}(t,z,x)$ for all $u\in\mc{U}$ and $(t,z,x)\in\mc{T}\times\mathbb{R}\times\mc{S}$, with equality if and only if $u$ is optimal for the respective control problem.
		
		Suppose that we have controls $\ul{u}$ and $\ol{u}$ such that
		\begin{equation*}
			F\big(t,\x_t,\ul{u}(t),\ul{\varphi}_1(t),\ul{\varphi}_2(t)\big)=\ul{H}\big(t,\x_t,\ul{\varphi}_1(t),\ul{\varphi}_2(t)\big)
		\end{equation*} 
		and
		\begin{equation*}
			F\big(t,\x_t,\ol{u}(t),\ol{\varphi}_1(t),\ol{\varphi}_2(t)\big)=\ol{H}\big(t,\x_t,\ol{\varphi}_1(t),\ol{\varphi}_2(t)\big).
		\end{equation*}
		It follows that $\big(\ul{V},\ul{\varphi}_1,\ul{\varphi}_2\big)$ and $\big(\ol{V},\ol{\varphi}_1,\ol{\varphi}_2\big)$ solve the BSDEs with drivers $F(\cdot,\cdot,\ul{u}(t),\cdot,\cdot)$ and 
		$F(\cdot,\cdot,\ol{u}(t),\cdot,\cdot)$, respectively. By uniqueness of solutions of BSDEs, 
		\begin{align*}
			\mc{J}(t,z,x,\ul{u})=\ul{V}(t,z,x)\quad\mbox{and}\quad \mc{J}(t,z,x,\ol{u})=\ol{V}(t,z,x).
		\end{align*}
		This proves that $\ul{u}$ and $\ol{u}$ are optimal controls.
		
		Conversely, suppose $\ul{u}$ and $\ol{u}$ are optimal. Then from Proposition \ref{J as solution to BSDE}, for some $\ul{\varphi}'_1$, $\ul{\varphi}'_2$, $\ol{\varphi}'_1$, and $\ol{\varphi}'_2$, the triples $(\mc{J}(t,z,x,\ul{u}),\ul{\varphi}'_1,\ul{\varphi}'_2)$ and $(\mc{J}(t,z,x,\ol{u}),\ol{\varphi}'_1,\ol{\varphi}'_2)$ are solutions to BSDEs with drivers $F(\cdot,\cdot,\ul{u}(t),\cdot,\cdot)$ and $F(\cdot,\cdot,\ol{u}(t),\cdot,\cdot)$, respectively. Furthermore, 
		\begin{equation*}
			F(t,x,u,\varphi_1,\varphi_2)\geq\ul{H}(t,x,\varphi_1,\varphi_2)\quad\mbox{and}\quad F(t,x,u,\varphi_1,\varphi_2)\leq\ol{H}(t,x,\varphi_1,\varphi_2).
		\end{equation*}
		Using the definition of $\ul{V}$ and $\ol{V}$ in Proposition \ref{V has version} and the comparison theorem, for all $s\in\mc{T}$ we have
		\begin{equation*}
			\mc{J}(0,z,x,\ul{u})=\ul{V}(0,z,x)\quad\mbox{and}\quad \mc{J}(0,z,x,\ol{u})=\ol{V}(0,z,x)
		\end{equation*}
		if and only if
		\begin{equation*}
			\mc{J}(s,z,x,\ul{u})=\ul{V}(s,z,x)\quad\mbox{and}\quad \mc{J}(s,z,x,\ol{u})=\ol{V}(s,z,x).
		\end{equation*}
		Since these processes have a unique canonical semimartingale decomposition, we have, up to indistinguishability,
		\begin{equation*}
			F\big(t,\x_t,\ul{u}(t),\ul{\varphi}'_1(t),\ul{\varphi}'_2(t)\big)=\ul{H}\big(t,\x_t,\ul{\varphi}_1(t),\ul{\varphi}_2(t)\big)
		\end{equation*}
		and
		\begin{equation*}
			F\big(t,\x_t,\ol{u}(t),\ol{\varphi}'_1(t),\ol{\varphi}'_2(t)\big)=\ol{H}\big(t,\x_t,\ol{\varphi}_1(t),\ol{\varphi}_2(t)\big).
		\end{equation*}
		Furthermore,
		\begin{equation*}
			\int_0^T(\ul{\varphi}_1'(t)dW_t+\ul{\varphi}_2'(t)d\ol{J}(t))=\int_0^T(\ul{\varphi}_1(t)dW_t+\ul{\varphi}_2(t)d\ol{J}(t))
		\end{equation*}
		and
		\begin{equation*}
			\int_0^T(\ol{\varphi}_1'(t)dW_t+\ol{\varphi}_2'(t)d\ol{J}(t))=\int_0^T(\ol{\varphi}_1(t)dW_t+\ol{\varphi}_2(t)d\ol{J}(t)).
		\end{equation*}
		Since $(\varphi_1,\varphi_2)\in M^2(\mc{T};\mathbb{R}^{K\times N},\mc{G}_t)\times P^2(\mc{T};\mathbb{R}^{K\times N},\mc{G}_t)$, then by the uniqueness of the martingale representation theorem, we have
		\begin{equation*}
			\|\ul{\varphi}_1-\ul{\varphi}'_1\|^2=0,\quad\|\ul{\varphi}_2-\ul{\varphi}'_2\|^2_{X_{t-}}=0,\quad\|\ol{\varphi}_1-\ol{\varphi}'_1\|=0,\quad\mbox{and}\quad\|\ol{\varphi}_2-\ol{\varphi}'_2\|^2_{X_{t-}}=0.
		\end{equation*}
		Since $F$ and $H$ are continuous with respect to their norms, then the result follows.
	\end{proof}
	The following proposition states that an optimal feedback control exists if an optimal control exists. Optimal feedback controls are controls that depend only on the current values of the state variables $(t,\z_t,\x_t)$.
	
	\begin{mypr}
		Suppose that there exist $\ul{u},\ol{u}\in\textsf{U}$ such that
		\begin{equation*}
			F(t,x,\ul{u},\varphi_1,\varphi_2)=\essinf_{u\in\textsf{U}}F(t,x,u,\varphi_1,\varphi_2)
		\end{equation*}
		and
		\begin{equation*}
			F(t,x,\ol{u},\varphi_1,\varphi_2)=\esssup_{u\in\textsf{U}}F(t,x,u,\varphi_1,\varphi_2).
		\end{equation*}
		Then there exist feedback controls $\ul{u}^*,\ol{u}^*:\mc{T}\times\mathbb{R}\times\mc{S}\mapsto\textsf{U}$ such that
		\begin{align*}
			\ul{u}^*(t,\z_t,\x_t)\quad\mbox{and}\quad \ol{u}^*(t,\z_t,\x_t)
		\end{align*}
		are optimal among all predictable controls.
	\end{mypr}
	\begin{proof}
		From Proposition \ref{minmax principle}, a control is optimal if and only if it minimizes or maximizes $F(t,x,u,\varphi_1,\varphi_2)$. Since $\varphi_1$ and $\varphi_2$ come from the solution of a Markovian BSDE, then $\varphi_1$ and $\varphi_2$ are Borel measurable functions. Using Filippov's Implicit Function Theorem, there are $\mc{B}(\mc{T}\times\mathbb{R}\times\mc{S})$-measurable maps $\ul{u}^*$ and $\ol{u}^*$ which minimizes and maximizes, respectively, $F(t,x,u,\varphi_1(t,z,x),\varphi_2(t,z,x))$ for all $(z,x)\in\mathbb{R}\times\mc{S}$ and almost all $t\in\mc{T}$.
	\end{proof}
	We now state the verification theorem.
	
	\begin{mypr}[Verification Theorem]\label{verification theorem}
		Define the integro-differential operator $\mc{L}$ by
		\begin{align*}
			\mc{L}v(t,z,x)&=f(t,z,x)v_z(t,z,x)+\frac{1}{2}\sigma^2(t,z,x)v_{zz}(t,z,x).	
		\end{align*}
		Write $\mathbf{v}:=(v(t,z,e_1),\ldots,v(t,z,e_N))^{\top}$. Consider the following Hamilton-Jacobi-Bellman (HJB) equations:
		\begin{equation}\label{hjbinf}
			\begin{cases}
				v_t(t,z,x)+\mc{L}v(t,z,x)+\ul{H}(t,x,v_z\sigma,\mathbf{v}+\eta(t,z))+\langle\mathbf{v}+\eta(t,z),Ax\rangle=0,\\
				v(T,z,x)=g(\z_T,\x_T)
			\end{cases}
		\end{equation}
		and
		\begin{equation}\label{hjbsup}
			\begin{cases}
				v_t(t,z,x)+\mc{L}v(t,z,x)+\ol{H}(t,x,v_z\sigma,\mathbf{v}+\eta(t,z))+\langle\mathbf{v}+\eta(t,z),Ax\rangle=0,\\
				v(T,z,x)=g(\z_T,\x_T),
			\end{cases}
		\end{equation}
		where $\eta(t,z):=(\eta(t,z,e_1),\ldots,\eta(t,z,e_N))^{\top}\in\mathbb{R}^N$.
		Suppose the HJB equations \eqref{hjbinf} and \eqref{hjbsup} admit $C^{1,2}([0,T)\times\mathbb{R})$ solutions $\ul{v}$ and $\ol{v}$, which satisfy the growth bound condition:
		\begin{align*}
			\|\ul{v}(s,z,x)\|^2+\|\ul{v}_z(s,z,x)\sigma(s,z,x)\|^2+\|\mathbf{v}+\eta(t,z)\|_{X_{t-}}^2\leq(1+|z|^2)
		\end{align*}
		and
		\begin{align*}
			\|\ol{v}(s,z,x)\|^2+\|\ol{v}_z(s,z,x)\sigma(s,z,x)\|^2+\|\mathbf{v}+\eta(t,z)\|_{X_{t-}}^2\leq(1+|z|^2).
		\end{align*}
		Then, for each $x\in\mc{S}$, the value functions $\ul{V}(t,z,x)=\ul{v}(t,z,x)$ and $\ol{V}(t,z,x)=\ol{v}(t,z,x)$ are the respective value functions of the control problems.
	\end{mypr}
	\begin{proof}
		We will prove the infimum case. Using Ito's formula,
		\begin{align*}
			d\ul{v}(t,z,x)&=\big[\ul{v}_t(t,z,x)+\mc{L}\ul{v}(t,z,x)+\langle\mathbf{\ul{v}}+\eta(t,z),Ax\rangle]dt+\ul{v}_z(t,z,x)\sigma(t,z,x)dW_t+\langle\mathbf{\ul{v}}+\eta(t,z),dM_t\rangle.
		\end{align*}
		Since $\ul{v}$ is the solution to \eqref{hjbinf}, then
		\begin{align*}
			d\ul{v}(t,z,x)&=-\ul{H}(t,x,v_z(t,z,x)\sigma(t,z,x),\mathbf{\ul{v}}+\eta(t,z))dt+\ul{v}_z(t,z,x)\sigma(t,z,x)dW_t+\langle\mathbf{\ul{v}}+\eta(t,z),dM_t\rangle.
		\end{align*}
		It follows that $\ul{v}$ solves \eqref{V inf}. Since $\ul{v}$ satisfies the growth bound condition, then by uniqueness $\ul{v}$ and $\ul{V}$ must agree, and similarly for $\varphi_1$ and $\varphi_2$ in their $M^2(\mc{T};\mathbb{R}^{K\times N},\mc{G}_t)$ and $P^2(\mc{T};\mathbb{R}^{K\times N},\mc{G}_t)$, respectively.
	\end{proof}
	Propositions \ref{V has version}, \ref{minmax principle}, and \ref{verification theorem} yield the following result.
	
	\begin{myth}
		The value processes $\ul{V}$ and $\ol{V}$ defined in \eqref{valinf} and \eqref{valsup} provide models of the bid and ask prices of the European call option.
	\end{myth}
	
	The result above can provide a connection to risk measures and nonlinear expectations. From \cite{ref2:conic,ref15:conic}, coherent risk measures can be represented as the suprema of expectations over a set of probability measures. Based on this result, the bid and ask prices can be represented as the infimum or supremum, respectively, of expectations over a convex set of probability measures, which is one of the main assumptions in conic finance \cite{ref1:conic}. 
	
	On the other hand, from \cite{ref7:sublinear}, nonlinear expectations can be represented as a solution to a BSDE if the driver satisfies some conditions. One type of nonlinear expectation is the sublinear expectation, which satisfies subadditity and positive homogeneity. Sublinear expectations can be represented as a supremum of a family of expectations \cite{ref2:sublinear,peng:sublinear}.

	\section{Conclusion} A pricing formula for a European call option whose underlying asset has dynamics governed by the state process is modelled using a system of partial differential equations. Modelling the bid and ask prices as the infimum and supremum, respectively, of the objective functional in \eqref{objfunc}, we have shown through a dynamic programming principle that these prices can be described as solutions to the BSDEs \eqref{V inf} and \eqref{V sup}. Optimality conditions for the drivers of the BSDEs are then proved in a minimum/maximum principle. Assuming some conditions are satisfied, we have also shown through a verification theorem that the bid and ask prices are solutions to the HJB equations \eqref{hjbinf} and \eqref{hjbsup}.
		
	Further work may be done to discuss model estimation and model calibration. Numerical methods may be used to obtain estimates of the semi-analytical European option price and the BSDEs associated to the stochastic control problems.

	\section*{Acknowledgements} The authors would like to thank the Australian Research Council and NSERC for continuing support.

	\bibliographystyle{siam}
	\bibliography{UpdatedReferences}

\end{document}